\documentclass[a4paper,aps,pra,showpacs,twocolumn,superscriptaddress]{revtex4-2} 

\usepackage[utf8]{inputenc}  
\usepackage[T1]{fontenc}     
\usepackage[american]{babel}  
\usepackage[svgnames,dvipsnames]{xcolor}
\usepackage[colorlinks=true,citecolor=DarkBlue,linkcolor=DarkBlue,urlcolor=DarkBlue,linktocpage,unicode]{hyperref}
\usepackage{amsmath,amssymb,amsthm,bm,mathtools,amsfonts,mathrsfs,bbm,dsfont} 
\usepackage{centernot}
\newtheorem{theorem}{Theorem}
\newtheorem{observation}[theorem]{Observation}

\newcommand{\tr}{{\mathrm{tr}}}

\newcommand{\eins}{\mathbbm{1}}

\renewcommand{\vr}{\ensuremath{\varrho}}
\renewcommand{\vec}[1]{\ensuremath{\boldsymbol{#1}}}
\usepackage{braket}
\usepackage[normalem]{ulem}
\usepackage{comment}

\begin{document}
\title{Semi-classical geometric tensor in multiparameter quantum information}

\author{Satoya Imai\hyperlink{email1}{\textsuperscript{*}}}
\affiliation{Institute of Systems and Information Engineering, University of Tsukuba, Tsukuba, Ibaraki 305-8573, Japan}
\affiliation{Center for Artificial Intelligence Research (C-AIR), University of Tsukuba, Tsukuba, Ibaraki 305-8577, Japan}
\affiliation{Istituto Nazionale di Ottica del Consiglio Nazionale delle Ricerche (CNR-INO), Largo Enrico Fermi 6, 50125, Firenze, Italy}
\affiliation{European Laboratory for Nonlinear Spectroscopy (LENS), Via N. Carrara 1, 50019 Sesto Fiorentino, Italy}

\author{Jing Yang\hyperlink{email2}{\textsuperscript{\textdagger}}}
\affiliation{Institute for Fundamental and Transdisciplinary Research,  Institute for Quantum Sensing, and Institute for Advanced Study in Physics, Zhejiang University, Hangzhou 310027, China}
\affiliation{Nordita, KTH Royal Institute of Technology and Stockholm University, Hannes Alfvéns vag 12, SE-106 91 Stockholm, Sweden}

\author{Luca Pezzè\hyperlink{email3}{\textsuperscript{\textdaggerdbl}}}
\affiliation{Istituto Nazionale di Ottica del Consiglio Nazionale delle Ricerche (CNR-INO), Largo Enrico Fermi 6, 50125, Firenze, Italy}
\affiliation{European Laboratory for Nonlinear Spectroscopy (LENS), Via N. Carrara 1, 50019 Sesto Fiorentino, Italy}

\date{\today}

\begin{abstract}
{The discrepancy between quantum distinguishability in Hilbert space and classical distinguishability in probability space is expressed by the gap between the quantum and classical Fisher information matrices (QFIM and CFIM, respectively). This intrinsic quantum obstruction is generally not saturable and plays a central role in both fundamental insights and practical applications in modern quantum physics. Here, we develop a geometrical framework for this gap by introducing the notion of semi-classical geometric tensor (SCGT). We relate this quantity to the quantum geometric tensor (QGT), whose real part equals the QFIM. We prove the matrix inequality between QGT and SCGT, which sharpens the standard inequality between QFIM and CFIM and provides novel multiparameter information bounds: the real part of the SCGT reproduces the CFIM plus an additional nonnegative contribution capturing quantum obstruction. This further motivates a natural extension of the Berry phase to the semi-classical setting.}
\end{abstract}

\maketitle

{\textit{Introduction.---}
Quantifying information in a quantum system requires distinguishing what is operationally accessible through measurements from what is intrinsic to the quantum state itself. This distinction clarifies fundamental limits on observables and has become central in quantum information~\cite{wilde2013quantum, hayashi2017quantum}, geometry~\cite{wootters1981statistical, braunstein1994statistical,bengtsson2017geometry}, and technology~\cite{giovannetti2011advances, pezze2018quantum}. The Fisher information provides a natural framework: the classical variant captures what a given measurement can reveal, while its quantum counterpart specifies the state’s ultimate limit, thereby making the measurement–state gap explicit.

The classical Fisher information matrix (CFIM) quantifies how much information a probability distribution carries about $m$ parameters $\vec{\theta} = (\theta_1, \ldots, \theta_m)$~\cite{helstrom1976quantum}. In quantum theory, probabilities arise from the Born rule: $p_\omega (\vec{\theta}) = \tr(\vr_{\vec{\theta}} E_\omega)$, where $\vr_{\vec{\theta}}$ is a $\vec{\theta}$-dependent quantum state and $\mathsf{E} = \{E_\omega\}$ is a positive operator-valued measure (POVM), satisfying $E_\omega \geq 0$ and $\sum_\omega E_\omega = \eins$. This foundational postulate yields the cornerstone inequality:
\begin{equation}
\label{eq:braunstein_caves}
  \mathcal{F}_C (\vr_{\vec{\theta}}, \mathsf{E})
  \leq 
  \mathcal{F}_Q (\vr_{\vec{\theta}}),
\end{equation}
where $\mathcal{F}_C$ and $\mathcal{F}_Q$ denote the CFIM and the quantum Fisher information matrix (QFIM), respectively~\cite{liu2020quantum}. This inequality holds for all states and POVMs, showing that the classical information extractable via measurements cannot exceed a fundamental bound determined solely by the quantum state.

In the single-parameter case ($m = 1$), equality in Eq.~(\ref{eq:braunstein_caves}) can always be saturated by an appropriate measurement~\cite{braunstein1994statistical}. In contrast, for multiple parameters ($m \geq 2$), saturation is generally impossible~\cite{amari2000methods,matsumoto2002new,pezze2017optimal,yang2019optimal}. The resulting gap between $\mathcal{F}_C$ and $\mathcal{F}_Q$ reflects \textit{measurement incompatibility}, which is rooted in the noncommutativity of the optimal measurements for different parameters~\cite{ holevo2011probabilistic, ragy2016compatibility, carollo2019quantumness,belliardo2021incompatibility}. This represents a genuinely quantum obstruction rather than a purely operational shortcoming.

So far, measurement incompatibility in the framework of Fisher information has primarily been explored from a quantum metrology perspective. Estimation-theoretic analyses typically compare the quantum Cramér-Rao bound~\cite{helstrom1976quantum}---a sensitivity limit set by $\mathcal{F}_Q^{-1}$---with the Holevo bound~\cite{holevo2011probabilistic}---a limit attainable with collective measurements on infinitely many copies of $\vr_{\vec{\theta}}$~\cite{demkowicz2020multi,albarelli2020perspective,pezze2025advances}. This viewpoint raises additional challenges, including the identification of optimal POVMs that saturate the Holevo bound and the characterization of performance gaps when only a finite number of measurements or copies is available (see, e.g., Refs.~\cite{conlon2021efficient, conlon2022gap, sidhu2021tight, hayashi2023tight}).

Distinguishability Riemannian metrics~\cite{bengtsson2017geometry}---Fisher-Rao (for probability distributions), Fubini-Study (pure states), and Bures (mixed states)---identify $\mathcal{F}_C$ and $\mathcal{F}_Q$ as statistical speeds under infinitesimal parameter shifts. However, these purely real metrics neglect the parameter-induced {\it phase twists} that accumulate as geometric (Berry) phases along closed loops~\cite{berry1984quantal,berry1989quantum}. In particular, the saturation in Eq.~(\ref{eq:braunstein_caves}) requires the absence of such twists (see below). A complete account of measurement incompatibility in the multiparameter regime thus demands a framework beyond real metric structures.

In this manuscript, we address the fundamental discrepancy between CFIM and QFIM from a quantum geometric perspective. We recognize that the \textit{quantum geometric tensor} (QGT)~\cite{provost1980riemannian,shapere1989geometric, carollo2020geometry}, $\mathcal{Q}(\vr_{\vec{\theta}})$, possesses nontrivial real and imaginary parts associated with the metric and the (Berry) curvature, respectively. We generalize the QGT by introducing the notion of \textit{semi-classical geometric tensor} (SCGT), $\mathcal{C}(\vr_{\vec{\theta}}, \mathsf{E})$, which is gauge invariant and measurement dependent. We establish the inequality
\begin{equation}
\label{eq:ourCQinequality}
  \mathcal{C}(\vr_{\vec{\theta}}, \mathsf{E})
  \leq
  \mathcal{Q}(\vr_{\vec{\theta}}),
\end{equation}
in close analogy with Eq.~(\ref{eq:braunstein_caves}). Taking the real parts of Eq.~(\ref{eq:ourCQinequality}) recovers the QFIM on the right-hand side, while the left-hand side provides a refined geometric account of the CFIM-QFIM gap, offering novel insights into measurement incompatibility. In the following, after formal definitions of the above quantities and the derivation of Eq.~(\ref{eq:ourCQinequality}), we discuss three implications: (i) identification of optimal measurements for which the CFIM-QFIM gap vanishes; (ii) derivation of tight bounds on measurement incompatibility; and (iii) elucidation of the geometric phase associated with the SCGT.

\vspace{1em}
\textit{Geometric tensors.---}
For a pure state $\ket{\psi_{\vec{\theta}}}$ parameterized by $\vec{\theta}$, the QGT is defined as the Hermitian, positive-semidefinite matrix $\mathcal{Q}(\ket{\psi_{\vec{\theta}}})$ with elements
\begin{equation} \label{eq:QGT}
    [\mathcal{Q}(\ket{\psi_{\vec{\theta}}})]_{ij}
    = 4 \braket{\partial_i \psi_{\vec{\theta}}|
    ( \eins - \ket{\psi_{\vec{\theta}}}\! \bra{\psi_{\vec{\theta}}} )
    |\partial_j \psi_{\vec{\theta}}},
\end{equation}
where $\ket{\partial_i \psi_{\vec{\theta}}} \equiv \partial_i \ket{\psi_{\vec{\theta}}}$ and $\partial_i \equiv {\partial}/{\partial \theta_i}$. A notable property of the QGT is its invariance under local gauge transformations: $\mathcal{Q}(\ket{\psi_{\vec{\theta}}^\prime}) = \mathcal{Q}(\ket{\psi_{\vec{\theta}}})$ for $\ket{\psi_{\vec{\theta}}^\prime} = e^{i \alpha_{\vec{\theta}}} \ket{\psi_{\vec{\theta}}}$, where the real quantity $\alpha_{\vec{\theta}}$ depends on $\vec{\theta}$. The real part of Eq.~(\ref{eq:QGT}) recovers the Fubini-Study metric~\cite{kobayashi1996foundations,bengtsson2017geometry}, while the imaginary part is the Berry curvature~\cite{berry1984quantal}. The QGT has probed quantum systems via fidelity susceptibility~\cite{zanardi2007information,campos2007quantum,gu2010fidelity,kolodrubetz2017geometry}, informed studies of quantum materials~\cite{yu2024quantum,jiang2025revealing}, been experimentally measured~\cite{tan2019experimental,gianfrate2020measurement,yu2020experimental,kang2024measurements}, and set the asymptotic conversion rate in the resource theory of asymmetry~\cite{yamaguchi2024quantum}.

For a general state $\vr_{\vec{\theta}}$, a generalization of the QGT is given by~\cite{carollo2020geometry}:
\begin{equation} \label{eq:mixedQGT}
    [\mathcal{Q}(\vr_{\vec{\theta}})]_{ij} = \tr (\vr_{\vec{\theta}} L_i L_j ).
\end{equation}
The Hermitian operator $L_i \equiv  L_i (\vr_{\vec{\theta}})$ is called the symmetric logarithmic derivative (SLD)~\cite{helstrom1976quantum}, and is defined by the relation $\partial_i \vr_{\vec{\theta}} = ( L_i \vr_{\vec{\theta}} + \vr_{\vec{\theta}} L_i )/2$. For pure states, Eq.~(\ref{eq:mixedQGT}) reduces to Eq.~(\ref{eq:QGT}), since $L_i(\ket{\psi_{\vec{\theta}}}) = 2 ( \ket{\partial_i \psi_{\vec{\theta}}}\! \bra{\psi_{\vec{\theta}}} + \ket{\psi_{\vec{\theta}}}\! \bra{\partial_i \psi_{\vec{\theta}}})$. The real part, ${\rm Re}[\mathcal{Q} (\vr_{\vec{\theta}})] \equiv \mathcal{F}_Q (\vr_{\vec{\theta}})$, is the QFIM with elements
\begin{equation} \label{eq:QFIM}
    [\mathcal{F}_Q(\vr_{\vec{\theta}})]_{ij}
    = \frac{1}{2} \tr[\vr_{\vec{\theta}} (L_i L_j + L_j L_i) ],
\end{equation}
while the imaginary part, ${\rm Im}[\mathcal{Q} (\vr_{\vec{\theta}})] \equiv \mathcal{G} (\vr_{\vec{\theta}})$~\cite{notation_Re_Im}, is known as mean Uhlmann curvature~\cite{uhlmann1986parallel,carollo2020geometry,carollo2018uhlmann}. Note that for single-parameter ($m = 1$) unitary encoding, the scalar quantum Fisher information has been related to entanglement properties in quantum systems~\cite{pezze2009entanglement,hyllus2012fisher,toth2012multipartite,pezze2016witnessing,imai2025metrological}, see also Refs.~\cite{gessner2018sensitivity,du2025quantifying} for investigations in the $m \geq 2$ scenario. 

Interestingly, Eq.~(\ref{eq:braunstein_caves}) is saturated under optimal measurement conditions, \textit{only if} the imaginary part of Eq.~(\ref{eq:mixedQGT}) vanishes~\cite{pezze2017optimal,yang2019optimal}: if $\mathcal{F}_C(\vr_{\vec{\theta}}, \mathsf{E}) = \mathcal{F}_Q(\vr_{\vec{\theta}})$, then $\mathcal{G}(\vr_{\vec{\theta}})=0$, where $\mathcal{F}_C (\vr_{\vec{\theta}}, \mathsf{E})$ is the CFIM with elements
\begin{equation} \label{eq.FIM}
[\mathcal{F}_C(\vr_{\vec{\theta}}, \mathsf{E})]_{i,j}
   = \sum_\omega \frac{
   [\partial_i p_\omega (\vec{\theta})]
   [\partial_j p_\omega (\vec{\theta})]
   }{p_\omega (\vec{\theta})}.
\end{equation}
This condition suggests that the CFIM can coincide with the QFIM only if the underlying Riemannian structure of the quantum state space has zero mean Uhlmann curvature. In other words, the parameter information encoded in the symplectic structure (described by the imaginary part of the QGT) cannot be accessed through the probabilities of measurement outcomes. Then it would be challenging to establish a geometric tensor that encompasses nontrivial real and imaginary parts for general POVM operators and fully recovers the QGT within appropriate limiting scenarios.

To address this puzzling discrepancy, we introduce the Hermitian, positive-semidefinite matrix with elements 
\begin{equation} \label{eq:def_matrixC}
    [\mathcal{C}(\vr_{\vec{\theta}}, \mathsf{E})]_{ij}
    = \sum_\omega
    \frac{[\chi_{\omega,i}(\vec{\theta})]^* \chi_{\omega,j}(\vec{\theta})}{p_\omega (\vec{\theta})},
\end{equation}
where $\chi_{\omega,i}(\vec{\theta}) \equiv \tr (\vr_{\vec{\theta}} E_\omega L_i)$ and $[\chi_{\omega,i}(\vec{\theta})]^* = \tr (\vr_{\vec{\theta}} L_i E_\omega)$ is the conjugation of $\chi_{\omega,i}(\vec{\theta})$. Due to the relation ${\rm Re}[\chi_{\omega,i}(\vec{\theta})] = [\partial_i p_\omega (\vec{\theta})]$, Eq.~(\ref{eq:def_matrixC}) provides a generalization of Eq.~(\ref{eq.FIM}) that includes a nontrivial imaginary part. In the main text, we focus on regular POVMs such that $p_{\omega} (\vec{\theta}) > 0$ for the sake of simplicity, while null POVMs such that $p_{\omega} (\vec{\theta}) = 0$ are discussed in Appendix~A.

For pure states, Eq.~(\ref{eq:def_matrixC}) has a structure analogue to Eq.~(\ref{eq:QGT}) (shown in Appendix~B):
\begin{equation}\label{eq:classical_GT}
    [\mathcal{C}(\ket{\psi_{\vec{\theta}}}, \mathsf{E})]_{ij}
    \!=\! 4 \braket{\partial_i \psi_{\vec{\theta}}|
    [\mathcal{M}(\ket{\psi_{\vec{\theta}}},\mathsf{E}) - \ket{\psi_{\vec{\theta}}}\! \bra{\psi_{\vec{\theta}}}]|\partial_j \psi_{\vec{\theta}}}, 
\end{equation}
where $\mathcal{M}(\ket{\psi_{\vec{\theta}}}, \mathsf{E}) \equiv \sum_\omega [1/p_\omega (\vec{\theta})] E_\omega \ket{\psi_{\vec{\theta}}}\! \bra{\psi_{\vec{\theta}}} E_\omega$. Also, Eq.~(\ref{eq:classical_GT}) shares with the QGT the property of gauge invariance: $\mathcal{C}(\ket{\psi_{\vec{\theta}}^\prime},\mathsf{E}) = \mathcal{C}(\ket{\psi_{\vec{\theta}}},\mathsf{E})$ for $\ket{\psi_{\vec{\theta}}^\prime} = e^{i \alpha_{\vec{\theta}}} \ket{\psi_{\vec{\theta}}}$ (shown in Appendix~C). The structural analogy with Eq.~(\ref{eq:QGT}) and the gauge invariance suggest calling Eq.~(\ref{eq:def_matrixC}) a {\it semi-classical geometric tensor} (SCGT), namely a counterpart of the QGT depending on the specific POVM. Finally, in Appendix~I, we show how the SCGT can be accessed for pure states by measuring suitable Hermitian operators.}

\vspace{1em}
\textit{Lower bound to the QGT.---}
Let us present and prove one of the main results of this manuscript:
\begin{observation} \label{ob:ourinequality}
    For a general state $\vr_{\vec{\theta}}$, consider Hermitian and positive-semidefinite matrices $\mathcal{Q}(\vr_{\vec{\theta}})$ in Eq.~(\ref{eq:mixedQGT}) and $\mathcal{C}(\vr_{\vec{\theta}},\mathsf{E})$ in Eq.~(\ref{eq:def_matrixC}). The inequality (\ref{eq:ourCQinequality}) holds for all states and POVMs, with the meaning $\vec{z}^\dagger \mathcal{C}(\vr_{\vec{\theta}},\mathsf{E}) \vec{z} \leq \vec{z}^\dagger \mathcal{Q}(\vr_{\vec{\theta}}) \vec{z}$ for any complex vector $\vec{z}$.
\end{observation}
\begin{proof}
    For any $\vec{z} \in \mathbb{C}^{m}$, we write
    \begin{equation} \label{eq:proofourCQinequality}
        \vec{z}^\dagger \mathcal{C}(\vr_{\vec{\theta}},\mathsf{E}) \vec{z}
        =\sum_{\omega} \frac{1}{p_\omega (\vec{\theta})}
        \lvert \tr (\vr_{\vec{\theta}} E_\omega \Tilde{L} ) \rvert^2,
    \end{equation}
    where $\Tilde{L}=\sum_i z_i L_i$. Letting $X = \sqrt{E_\omega} \sqrt{\vr_{\vec{\theta}}}$ and $Y = \sqrt{E_\omega} \Tilde{L} \sqrt{\vr_{\vec{\theta}}}$ and applying the Cauchy-Schwarz inequality $|\tr(X^\dagger Y)|^2 \leq \tr(XX^\dagger) \tr(YY^\dagger)$ yields that $\lvert \tr (\vr_{\vec{\theta}} E_\omega \Tilde{L} ) \rvert^2 \leq p_\omega (\vec{\theta}) \tr ( E_\omega \Tilde{L} \vr_{\vec{\theta}} \Tilde{L}^\dagger)$. Inserting this into Eq.~(\ref{eq:proofourCQinequality}) and using $\sum_\omega E_\omega = \eins$, we obtain Eq.~(\ref{eq:ourCQinequality}). Since $\lvert \tr (\vr_{\vec{\theta}} E_\omega \Tilde{L} ) \rvert^2  \geq  0$, we directly obtain $\mathcal{C}(\vr_{\vec{\theta}},\mathsf{E}) \geq 0$.
\end{proof}

Let us discuss the saturation of the inequality~(\ref{eq:ourCQinequality}). For pure states, $\mathcal{C}(\ket{\psi_{\vec{\theta}}}, \mathsf{E}) = \mathcal{Q}(\ket{\psi_{\vec{\theta}}})$ holds for \textit{every} rank-one POVM $\{E_\omega = \ket{\pi_\omega}\! \bra{\pi_\omega}\}$, where $E_\omega$ is not necessarily projective (namely, $E_\omega E_{\omega^\prime} = \delta_{\omega, \omega^\prime} E_\omega$ does not necessarily hold). This can be seen by noticing that $E_\omega \ket{\psi_{\vec{\theta}}}\! \bra{\psi_{\vec{\theta}}} E_\omega = p_\omega (\vec{\theta}) E_\omega$ for $E_{\omega} = \ket{\pi_{\omega}}\! \bra{\pi_\omega}$ and thus $\mathcal{M}(\ket{\psi_{\vec{\theta}}}, \mathsf{E})$ in Eq.~(\ref{eq:classical_GT}) becomes the identity matrix for any $\ket{\psi_{\vec{\theta}}}$. The consequence of this saturation will be elaborated in the next section.

For general mixed states and regular POVM operators, $\mathcal{C}(\vr_{\vec{\theta}}, \mathsf{E}) = \mathcal{Q}(\vr_{\vec{\theta}})$ holds if and only if there exists a rank-one POVM $\{E_\omega = \ket{\pi_\omega}\! \bra{\pi_\omega}\}$ such that
\begin{equation}\label{eq:R_case}
    \bra{\pi_\omega} \otimes \bra{\pi_\omega} (L_i \otimes \eins - \eins \otimes L_i) \ket{\psi_{x, \vec{\theta}}} \otimes \ket{\psi_{y, \vec{\theta}}}
        = 0,
\end{equation}
holds for all $i,\omega,x,y$, where $\ket{\psi_{x, \vec{\theta}}}$ is the eigenstate of $\vr_{\vec{\theta}}$. It is straightforward to see that Eq.~(\ref{eq:R_case}) is verified for pure states since $\ket{\psi_{x, \vec{\theta}}} = \ket{\psi_{y, \vec{\theta}}} = \ket{\psi_{\vec{\theta}}}$. The proof of Eq.~(\ref{eq:R_case}) is shown in Appendices~D~and~E. The saturation condition in null POVMs is discussed in Appendix~F.

\vspace{1em}
\textit{Tighter lower bound to the QFIM.---}
Let us decompose the SCGT into real and imaginary parts: ${\rm Re} [\mathcal{C}(\vr_{\vec{\theta}}, \mathsf{E})] \equiv \mathcal{F}_C(\vr_{\vec{\theta}}, \mathsf{E}) + \mathcal{I}(\vr_{\vec{\theta}}, \mathsf{E})$ and ${\rm Im} [\mathcal{C}(\vr_{\vec{\theta}}, \mathsf{E})] \equiv \mathcal{D}(\vr_{\vec{\theta}}, \mathsf{E})$, where $\mathcal{I}(\vr_{\vec{\theta}}, \mathsf{E})$ and $\mathcal{D}(\vr_{\vec{\theta}}, \mathsf{E})$ have elements
\begin{subequations}
    \begin{align} 
        \label{eq:classical_I}
        [\mathcal{I}(\vr_{\vec{\theta}}, \mathsf{E})]_{ij}
        &= \sum_{\omega} 
        \frac{{\rm Im}[\chi_{\omega,i}(\vec{\theta})] {\rm Im}[\chi_{\omega,j}(\vec{\theta})]} {p_\omega (\vec{\theta})},
        \\ \label{eq:classical_D}
        [\mathcal{D}(\vr_{\vec{\theta}}, \mathsf{E})]_{ij}
        &=\sum_{\omega}
        \frac{\xi_{\omega, ij}(\vec{\theta}) - \xi_{\omega, ji}(\vec{\theta})} {p_\omega (\vec{\theta})},
        \\
        \xi_{\omega, ij}(\vec{\theta})
        &\equiv {\rm Re}[\chi_{\omega,i}(\vec{\theta})] {\rm Im}[\chi_{\omega,j}(\vec{\theta})]. \label{eq.:xi}
    \end{align}
\end{subequations}
In general, $\mathcal{I}(\vr_{\vec{\theta}}, \mathsf{E})$ and $\mathcal{D}(\vr_{\vec{\theta}}, \mathsf{E})$ are nonzero matrices, while $\mathcal{D}(\vr_{\vec{\theta}}, \mathsf{E}) = 0$ holds in the single-parameter case. For more expressions of pure states and unitary transformations, see Appendix~G.

We can present our second main result:
\begin{observation}\label{ob:generalization_Braunstein-Caves}
    We have that $\mathcal{F}_C(\vr_{\vec{\theta}},\mathsf{E}) + \mathcal{I}(\vr_{\vec{\theta}},\mathsf{E})$ provides a tighter lower bound to $\mathcal{F}_Q(\vr_{\vec{\theta}})$ than $\mathcal{F}_C(\vr_{\vec{\theta}},\mathsf{E})$:
    \begin{equation} \label{eq:lowerbound_FC_FQ}
        \mathcal{F}_C(\vr_{\vec{\theta}},\mathsf{E})
        \leq \mathcal{F}_C(\vr_{\vec{\theta}},\mathsf{E})
        + \mathcal{I}(\vr_{\vec{\theta}},\mathsf{E})
        \leq \mathcal{F}_Q(\vr_{\vec{\theta}}).
    \end{equation}
    \end{observation}
\begin{proof}
    Recall that the transpose of a positive-semidefinite matrix $X \! \geq \! 0$ is also positive-semidefinite $X^\top \! \geq  \! 0$. Therefore, ${\rm Re}[X] \!  \geq  \! 0$. Taking $X = \mathcal{Q}(\vr_{\vec{\theta}}) - \mathcal{C}(\vr_{\vec{\theta}},\mathsf{E}) \geq 0$ in Eq.~(\ref{eq:ourCQinequality}), we obtain the right-hand inequality of Eq.~(\ref{eq:lowerbound_FC_FQ}). The left-hand inequality of Eq.~(\ref{eq:lowerbound_FC_FQ}) follows from that $\mathcal{F}_C(\vr_{\vec{\theta}},\mathsf{E})\geq 0$ by definition and that $\mathcal{I}(\vr_{\vec{\theta}},\mathsf{E}) \geq 0$ since $\vec{z}^\dagger \mathcal{I}(\vr_{\vec{\theta}},\mathsf{E}) \vec{z} = \sum_\omega  {\rm Im}[\tr (\vr_{\vec{\theta}} E_\omega \Tilde{L} )]^2/p_\omega (\vec{\theta}) \geq 0$ for any vector $\vec{z} \in \mathbb{C}^{m}$, with $\Tilde{L}=\sum_i z_i L_i$.
\end{proof}

We have several remarks on Observation~\ref{ob:generalization_Braunstein-Caves}. First, Eq.~(\ref{eq:lowerbound_FC_FQ}) is the generalization of Eq.~(\ref{eq:braunstein_caves}), originally derived by Braunstein and Caves in the single-parameter case~\cite{braunstein1994statistical} and later extended to multiparameter cases~\cite{pezze2017optimal,yang2019optimal}. The additional term $\mathcal{I}(\vr_{\vec{\theta}},\mathsf{E})$ quantifies a nontrivial gap between the CFIM and the QFIM.

For pure states, the gap is tight: $\mathcal{F}_C(\ket{\psi_{\vec{\theta}}},\mathsf{E}) + \mathcal{I}(\ket{\psi_{\vec{\theta}}},\mathsf{E}) = \mathcal{F}_Q(\ket{\psi_{\vec{\theta}}})$ holds for \textit{any} rank-one POVM, since $\mathcal{C}(\ket{\psi_{\vec{\theta}}},\mathsf{E}) =\mathcal{Q} (\ket{\psi_{\vec{\theta}}})$ as discussed above. Then, the quantity $\mathcal{I}(\ket{\psi_{\vec{\theta}}},\mathsf{E})$ precisely quantifies the difference between the QFIM and the CFIM. The necessary and sufficient condition for $\mathcal{I}(\ket{\psi_{\vec{\theta}}},\mathsf{E})=0$ is ${\rm Im}[\chi_{\omega,i}(\vec{\theta})] = 0$ for all $i$ and $\omega$. This recovers the necessary and sufficient condition for the existence of a rank-one regular POVM to achieve $\mathcal{F}_C(\ket{\psi_{\vec{\theta}}},\mathsf{E}) = \mathcal{F}_Q(\ket{\psi_{\vec{\theta}}})$, as introduced in Ref.~\cite{pezze2017optimal}, see Appendix~H for more details.

For general mixed states, the necessary and sufficient condition for $\mathcal{F}_C(\vr_{\vec{\theta}},\mathsf{E}) + \mathcal{I}(\vr_{\vec{\theta}},\mathsf{E}) = \mathcal{F}_Q(\vr_{\vec{\theta}})$ is given in Eq.~(\ref{eq:R_case}). We have that $\mathcal{I}(\vr_{\vec{\theta}},\mathsf{E})=0$ if and only if ${\rm Im}[\chi_{\omega,i}(\vec{\theta})] = 0$ for all $i$ and $\omega$. This recovers the necessary and sufficient condition for the existence of a rank-one regular POVM to achieve $\mathcal{F}_C(\vr_{\vec{\theta}},\mathsf{E}) = \mathcal{F}_Q(\vr_{\vec{\theta}})$, discussed in Refs.~\cite{yang2019optimal,yang2026}, see Appendix~H for more details.

In the single-parameter case ($m=1$), Eq.~(\ref{eq:R_case}) becomes $\braket{\pi_\omega|L|\psi_{x,\theta}} \! \braket{\pi_\omega|\psi_{y,\theta}} = \braket{\pi_\omega|\psi_{x,\theta}} \! \braket{\pi_\omega|L|\psi_{y,\theta}}$. This condition is satisfied for all $x,y$ by choosing $\ket{\pi_\omega}$ as an eigenstate of the SLD operator $L$. Such a choice also ensures that ${\rm Im}[\chi_{\omega,i}(\theta)] = 0$, given that the eigenvalues of $L$ are real. We thus recover that $\mathcal{F}_C (\vr_{{\theta}},\mathsf{E}) \leq \mathcal{F}_Q (\vr_{{\theta}})$ can always be saturated~\cite{braunstein1994statistical}, where both quantities are scalars. 

Finally, for the trace of the QFIM, we have the additional lower bound:
\begin{equation}
    \label{eq:trace_version_FC_FG}
            \lVert \Delta(\vr_{\vec{\theta}}, \mathsf{E}) \rVert_{\rm tr}
            \!+\! \tr [\mathcal{F}_C(\vr_{\vec{\theta}},\mathsf{E})
            \!+\! \mathcal{I}(\vr_{\vec{\theta}},\mathsf{E})]
            \! \leq \! 
            \tr[\mathcal{F}_Q(\vr_{\vec{\theta}})],
\end{equation}
where $\Delta(\vr_{\vec{\theta}},\mathsf{E}) \equiv \mathcal{G}(\vr_{\vec{\theta}}) - \mathcal{D}(\vr_{\vec{\theta}},\mathsf{E})$ and $\lVert X \rVert_{\rm tr} \equiv \sum_{i} \lvert x_i \rvert$ denotes the trace norm for $x_i$ being the eigenvalues of a matrix $X$. The derivation of Eq.~(\ref{eq:trace_version_FC_FG}) is based on the Belavkin-Grishanin inequality~\cite{belavkin1973optimum} (see Lemma~1 in Ref.~\cite{tsang2020quantum}): for a positive-semidefinite matrix $X \geq 0$, it holds that $\tr[ {\rm Re}(X) ] \geq \lVert {\rm Im}(X) \rVert_{\rm tr}$. Taking $X = \mathcal{Q}(\vr_{\vec{\theta}}) - \mathcal{C}(\vr_{\vec{\theta}},\mathsf{E}) \geq 0$ directly yields Eq.~(\ref{eq:trace_version_FC_FG}). We note that Eq.~(\ref{eq:trace_version_FC_FG}) provides a tighter lower bound than the one obtained by taking the trace of $\mathcal{F}_C(\vr_{\vec{\theta}},\mathsf{E}) + \mathcal{I}(\vr_{\vec{\theta}},\mathsf{E}) \leq \mathcal{F}_Q(\vr_{\vec{\theta}})$. The difference between the imaginary parts $\mathcal{G}(\vr_{\vec{\theta}})$ and $\mathcal{D}(\vr_{\vec{\theta}},\mathsf{E})$, which cannot appear in the matrix inequality~(\ref{eq:lowerbound_FC_FQ}), emerges as an additional term in the scalar inequality~(\ref{eq:trace_version_FC_FG}).

\vspace{1em}
\textit{Closeness between CFIM and QFIM.---}
Besides the matrix inequality~(\ref{eq:lowerbound_FC_FQ}), we provide a scalar bound to characterize further how close the CFIM is to the QFIM for given POVM operators.
\begin{observation}\label{ob:scalar_generalization_Braunstein-Caves}
    For a given Hermitian and positive-definite matrix $W$ with $\tr(W) = 1$ (without loss of generality), it holds that
    \begin{equation} \label{eq:scalar_lowerbound_FC_FQ}
    \tr \big(
    W \mathcal{F}_Q^{-\frac{1}{2}} \mathcal{F}_C \mathcal{F}_Q^{-\frac{1}{2}}
    \big) \leq 1 - \Gamma_{W},
    \end{equation}
    where 
    \begin{equation} \label{eq:upperbound_gamma}
        \Gamma_{W} \!=\!
        \tr \big(
        W \mathcal{F}_Q^{-\frac{1}{2}} \mathcal{I} \mathcal{F}_Q^{-\frac{1}{2}}
        \big)
        +\lVert \sqrt{W} \mathcal{F}_Q^{-\frac{1}{2}} \Delta \mathcal{F}_Q^{-\frac{1}{2}} \sqrt{W} \rVert_{\rm tr}
    \end{equation}
    is a non-negative quantity.
\end{observation}
\begin{proof}
    We take $X = \sqrt{W} \mathcal{F}_Q^{-\frac{1}{2}} (\mathcal{Q} - \mathcal{C} ) \mathcal{F}_Q^{-\frac{1}{2}} \sqrt{W} \geq 0$ due to Eq.~(\ref{eq:ourCQinequality}) and $\mathcal{F}_Q \geq 0$. We obtain Eq.~(\ref{eq:scalar_lowerbound_FC_FQ}) by following the Belavkin-Grishanin inequality (as discussed above) and noting that $\tr[ {\rm Re}(X) ] = 1 - \tr[W  \mathcal{F}_Q^{-\frac{1}{2}} (\mathcal{F}_C + \mathcal{I})  \mathcal{F}_Q^{-\frac{1}{2}}]$.
\end{proof}

We notice that Eq.~(\ref{eq:scalar_lowerbound_FC_FQ}) yields a tighter bound than $\tr (W \mathcal{F}_Q^{-\frac{1}{2}} \mathcal{F}_C \mathcal{F}_Q^{-\frac{1}{2}} ) \leq 1$, which can be obtained from $\mathcal{F}_C \leq \mathcal{F}_Q$. The upper bound of Eq.~(\ref{eq:scalar_lowerbound_FC_FQ}) is computable, since $\Gamma_{W}$ depends on the specific POVM. In the case of $W = \eins/m$ for $m$ being the number of parameters, Eq.~(\ref{eq:scalar_lowerbound_FC_FQ}) reduces to
\begin{equation} \label{eq:close_qCFIMs}
    \tr \big(
    \mathcal{F}_Q^{-1} \mathcal{F}_C \big)
    \leq m - \Gamma_{\eins/m},
\end{equation}
where $\Gamma_{\eins/m} = \tr (\mathcal{F}_Q^{-1} \mathcal{I} ) + \lVert \mathcal{F}_Q^{-\frac{1}{2}} \Delta \mathcal{F}_Q^{-\frac{1}{2}} \rVert_{\rm tr} \in [0, m]$. Eq.~(\ref{eq:close_qCFIMs}) is related with the inequality $\tr (\mathcal{F}_Q^{-1} \mathcal{F}_C ) \leq d - 1$, derived by Gill and Massar~\cite{gill2000state} {(also see \cite{wang2025tight})}, where $d$ is the dimension of the Hilbert space of $\vr_{\vec{\theta}}$. Our upper bound in Eq.~(\ref{eq:close_qCFIMs}) is tighter than the Gill-Massar bound in the generally relevant case of large $d$ (e.g., $d=2^N$ for $N$ qubits) and relatively small $m$.

Finally, we remark that the quantity $\mathcal{R} \equiv \lVert i \mathcal{F}_Q^{-1} \mathcal{G} \rVert_\infty \in [0,1]$ has been considered to characterize measurement incompatibility in multiparameter quantum estimation~\cite{carollo2019quantumness} (see also Refs.~\cite{razavian2020quantumness,belliardo2021incompatibility}), where $\lVert X \rVert_\infty$ is the largest absolute eigenvalue of $X$ (different notion of measurement incompatibility as the absence of joint measurability has also been discussed in quantum information, see~\cite{heinosaari2016invitation,guhne2023colloquium}). In multiparameter quantum metrology, the quantity $\mathcal{R}$ provides an upper bound of the ratio between the Holevo bound~\cite{holevo2011probabilistic} and the Helstrom Cramér-Rao bound~\cite{helstrom1976quantum}, see Refs.~\cite{albarelli2020perspective,demkowicz2020multi,pezze2025advances} for more details. Based on Eq.~(\ref{eq:ourCQinequality}), we can present
\begin{equation} \label{eq.Rineq}
    \lVert \mathcal{F}_Q^{-1} \mathcal{C} - \eins \rVert_\infty
    \leq \mathcal{R} \leq
    \lVert \eins - \mathcal{F}_Q^{-\frac{1}{2}} \mathcal{C} \mathcal{F}_Q^{-\frac{1}{2}} \rVert_\infty,
\end{equation}
where $\lVert \eins - \mathcal{F}_Q^{-\frac{1}{2}} \mathcal{C} \mathcal{F}_Q^{-\frac{1}{2}} \rVert_\infty \leq 1$. If $\mathcal{C} = \mathcal{Q}$, then both inequalities become equalities. 

The left-hand inequality in Eq.~(\ref{eq.Rineq}) is derived by using $\mathcal{C} \leq \mathcal{Q}$ and $\mathcal{Q} = \mathcal{F}_Q + i \mathcal{G}$. To obtain the right-hand inequality, we use $X^{-\frac{1}{2}} Y X^{-\frac{1}{2}} \geq 0$, valid for positive-semidefinite matrices $X, Y$: Taking $X=\mathcal{F}_Q$ and $Y = \mathcal{Q} - \mathcal{C}$, we obtain $- i \mathcal{F}_Q^{-\frac{1}{2}} \mathcal{G} \mathcal{F}_Q^{-\frac{1}{2}} \leq \eins - \mathcal{F}_Q^{-\frac{1}{2}} \mathcal{C} \mathcal{F}_Q^{-\frac{1}{2}} \leq \eins$, where $\lVert - i \mathcal{F}_Q^{-\frac{1}{2}} \mathcal{G} \mathcal{F}_Q^{-\frac{1}{2}} \rVert_\infty = \lVert i \mathcal{F}_Q^{-1} \mathcal{G} \rVert_\infty$.

\vspace{1em}
\textit{Imaginary part of the SCGT.---}
For pure states, the imaginary part of the QGT is reformulated as
\begin{equation} \label{eq:mathG_pure_another}
    \mathcal{G}(\ket{\psi_{\vec{\theta}}}) = - 2 \Omega (\ket{\psi_{\vec{\theta}}}),
\end{equation}
where $\Omega (\ket{\psi_{\vec{\theta}}})$ is the Berry curvature, with elements $\Omega_{ij} = \partial_i \mathcal{A}_j - \partial_j \mathcal{A}_i$, and $\mathcal{A} \equiv \mathcal{A}(\ket{\psi_{\vec{\theta}}})$ is the Berry connection with $\mathcal{A}_j \equiv i \braket{\psi_{\vec{\theta}}|\partial_j \psi_{\vec{\theta}}}$~\cite{berry1984quantal,simon1983holonomy,aharonov1987phase}. The form of Eq.~(\ref{eq:mathG_pure_another}) can be checked by $[\mathcal{G}(\ket{\psi_{\vec{\theta}}})]_{ij} = 4 {\rm Im}[\braket{\partial_i \psi_{\vec{\theta}}|\partial_j \psi_{\vec{\theta}}}]$. The Berry curvature describes an effective gauge field strength in the parameter space, analogous to a fictitious magnetic field experienced during adiabatic evolution~\cite{sakurai2020modern}.

Let us write $\mathcal{A} = \sum_{\omega} \mathcal{A}_\omega$ with $[\mathcal{A}_\omega]_j \equiv i \braket{\psi_{\vec{\theta}}|E_\omega|\partial_j \psi_{\vec{\theta}}}$, where $\mathcal{A}$ is a real vector but $\mathcal{A}_\omega$ is not a real vector due to $[\mathcal{A}_\omega]_j^* = [\mathcal{A}_\omega]_j - i \partial_j p_\omega (\vec{\theta})$. Since the imaginary part of the SCGT is $[\mathcal{D}(\ket{\psi_{\vec{\theta}}},\mathsf{E})]_{ij} = 4 {\rm Im}[\braket{\partial_i \psi_{\vec{\theta}}|\mathcal{M}(\ket{\psi_{\vec{\theta}}},\mathsf{E})|\partial_j \psi_{\vec{\theta}}}]$, a direct calculation leads to
\begin{equation} \label{eq:mathD_pure_another}
    \mathcal{D}(\ket{\psi_{\vec{\theta}}},\mathsf{E})
    = -2i \sum_{\omega} \frac{
    \mathcal{A}_\omega^* \mathcal{A}_\omega^\top
    - \mathcal{A}_\omega \mathcal{A}_\omega^\dagger}{p_\omega (\vec{\theta})},
\end{equation}
where $*$, $\top$, and $\dagger$ respectively denote the complex conjugation, the transposition, and the Hermitian (conjugate transpose). Using Eqs.~(\ref{eq:mathG_pure_another}, \ref{eq:mathD_pure_another}) and letting $\Omega = \sum_\omega \Omega_\omega$ with elements $[\Omega_\omega]_{ij} \equiv \partial_i [\mathcal{A}_\omega]_j - \partial_j [\mathcal{A}_\omega]_i$, we can express the gap as $\Delta \equiv \mathcal{G}-\mathcal{D}  = \sum_\omega \Delta_\omega$, where $\Delta_\omega$ vanishes for a rank-one POVM.

The integral of $\Omega (\ket{\psi_{\vec{\theta}}})$ over an oriented manifold $S$ in the parameter space is known as the Berry phase~\cite{berry1984quantal}:
\begin{equation} \label{eq:berryphase}
    \phi_Q \equiv \frac{1}{2} \int_S
    \sum_{i,j} \, [\Omega (\ket{\psi_{\vec{\theta}}})]_{ij}
    \,
    d\theta_i \wedge d\theta_j,
\end{equation}
where $\wedge$ is the wedge (or exterior) product and $d\theta_i \wedge d\theta_j$ is an area element on $S$~\cite{nakahara2018geometry}. In particular, in the two-dimensional parameter space ($m=2$), the Gauss–Bonnet theorem states that $\nu_Q = \phi_Q/(2 \pi)$ is always an integer, known as the first Chern number, which serves as a topological invariant~\cite{xiao2010berry,qi2011topological,bernevig2013topological}.

In analogy to Eq.~(\ref{eq:berryphase}), we can introduce
\begin{equation} \label{eq:SCGT_phase}
    \phi_C
    \equiv
    - \frac{1}{4} \int_S
    \sum_{i,j} \, [\mathcal{D}(\ket{\psi_{\vec{\theta}}}, \mathsf{E})]_{ij}
    \,
    d\theta_i \wedge d\theta_j.
\end{equation}
We have that $\phi_C=\phi_Q$ for any rank-one POVM, but $\phi_C \neq \phi_Q$ for general POVMs. Thus, $\nu_C \equiv \phi_C/(2\pi)$ cannot always be an integer, because the Gauss–Bonnet theorem cannot be applied {due to the form of Eq.~(\ref{eq:mathD_pure_another}).}

\vspace{1em}
\textit{Example.---}
Consider a single-qubit state with $\vec{\theta} = (\vartheta, \varphi)$ for the intervals $\vartheta \in [0, \pi]$ and $\varphi \in [0, 2 \pi]$: $\ket{\psi_{\vec{\theta}}} = \sin(\vartheta/2) \ket{0} + e^{i \varphi} \cos(\vartheta/2) \ket{1}$, where $\ket{0}$ and $\ket{1}$ are the eigenstates of the Pauli-$z$ matrix with $\pm 1$ eigenvalues, respectively. Taking the non-rank-one POVM with two outcomes $\mathsf{E} = \{E_\omega = \varepsilon \ket{\omega}\! \bra{\omega} + (1-\varepsilon) \eins/ 2\}$ for $\omega = 0,1$ and a parameter $\varepsilon \in [0,1]$, we have $\chi_{1, \vartheta} = - \chi_{2, \vartheta} = \varepsilon \sin(\vartheta)/2$ and $\chi_{1, \varphi} = - \chi_{2, \varphi} = - i \varepsilon \sin^2(\vartheta)/2$. Thus, 
\begin{equation}
    \mathcal{Q} =
    \begin{bmatrix}
    1 & -i \sin(\vartheta)
    \\
    i \sin(\vartheta) & \sin^2(\vartheta)
    \end{bmatrix},
    \quad
    \mathcal{C}
    = f_{\varepsilon,\vartheta} \mathcal{Q}
\end{equation}
where $f_{\varepsilon,\vartheta} = \varepsilon^2 \sin^2(\vartheta)/[1-\varepsilon^2 \cos^2(\vartheta)]$. We obtain $\nu_Q = 1$ and $\nu_C (\varepsilon) = 1 - [(1/\varepsilon) - \varepsilon] {\rm arctanh} (\varepsilon) \in [0,1]$, where $\nu_C (\varepsilon)$ monotonically increases for $\varepsilon$. This suggests that $\nu_C$ may provide a lower bound to $\nu_Q$ in general.

\vspace{1em}
\textit{Conclusion.---}
{We have examined the gap between QFIM and CFIM from a geometric perspective by introducing the SCGT and establishing an inequality between QGT and SCGT that sharpens the known discrepancy between QFIM and CFIM. Based on this approach, we derived multiparameter quantum information bounds.}

Our results open several avenues for further research. First, our findings may advance toward the characterization of measurement incompatibility and the saturation problem of the quantum Cramér-Rao bound in multiparameter metrology~\cite{albarelli2020perspective,pezze2025advances}, recognized as a relevant open problem in quantum information~\cite{horodecki2022five}. Second, exploring the role of the SCGT or its real part could provide fresh insights into quantum information science, such as the theory of asymmetry~\cite{yamaguchi2024quantum} and operational frameworks based on the quantum Fisher information in thermodynamics~\cite{marvian2022operational} and quantum resource theories~\cite{tan2021fisher}. Moreover, our results may be extended beyond SLD operators, and be related to generalized quantum speeds~\cite{petz2002covariance,petz2011introduction}, susceptibilities~\cite{hauke2016measuring}, Gaussian states~\cite{jiang2014quantum,nichols2018multiparameter,oh2019optimal}, and higher-order geometric quantities~\cite{hetenyi2023fluctuations}. Finally, beyond the theoretical interests of our findings, {the practical accessibility of the SCGT is discussed in Appendix~I.}

\vspace{1em}
{\it Acknowledgments.---}
We thank
Francesco Albarelli,
Leonardo Banchi,
Hongzhen Chen,
Balázs Hetény,
Yutaka Shikano,
Augusto Smerzi,
Hiroyasu Tajima,
Lingna Wang,
Frank Wilczek,
Benjamin Yadin,
and Haidong Yuan,
for discussions.
S.I. acknowledges support from Horizon Europe programme HORIZON-CL4-2022-QUANTUM-02-SGA via the project 101113690 (PASQuanS2.1) and JST ASPIRE (JPMJAP2339).
J.Y. acknowledges support from Zhejiang University start-up grants, Zhejiang Key Laboratory of R\&D and Application of Cutting-edge Scientific Instruments, and Wallenberg Initiative on Networks and Quantum Information (WINQ).
L.P. acknowledges support from the QuantERA project SQUEIS (Squeezing enhanced inertial sensing), funded by the European Union’s Horizon Europe Program and the Agence Nationale de la Recherche (ANR-22-QUA2-0006). 




\vspace{1em}
{\footnotesize
\hypertarget{email1}{}\noindent\textsuperscript{*} \href{mailto:satoyaimai@yahoo.co.jp}{satoyaimai@yahoo.co.jp} \\
\hypertarget{email2}{}\textsuperscript{\textdagger} \href{mailto:jing.yang.quantum@zju.edu.cn}{jing.yang.quantum@zju.edu.cn} \\
\hypertarget{email3}{}\textsuperscript{\textdaggerdbl} \href{mailto:luca.pezze@ino.cnr.it}{luca.pezze@ino.cnr.it}
}
\clearpage
\newpage
\section*{END MATTER}
\paragraph*{Appendix~A: Extension of Eq.~(\ref{eq:def_matrixC}) to the general POVM case.---}
Let us write $\mathcal{C}(\vr_{\vec{\theta}},\mathsf{E}) = \sum_{\omega} \mathcal{C}_\omega(\vr_{\vec{\theta}})$, where $\mathcal{C}_\omega(\vr_{\vec{\theta}})$ has elements 
\begin{equation} \label{eq:SCGR_omega_level}
    [\mathcal{C}_\omega(\vr_{\vec{\theta}})]_{ij} \!=\! 
        \begin{dcases}
        \frac{ [\chi_{\omega,i}(\vec{\theta})]^* \chi_{\omega,j}(\vec{\theta}) }{ p_\omega (\vec{\theta}) }
        \!
        &{\rm for} \,\, {\rm Regular}\\ 
        \lim_{\vec{\Tilde{\theta}} \to \vec{\theta}}
        \frac{[{\chi}_{\omega,i}(\vec{\Tilde{\theta}})]^* {\chi}_{\omega,j}(\vec{\Tilde{\theta}})}{p_\omega(\vec{\Tilde{\theta}})}
        \!
        &{\rm for} \,\, {\rm Null}\\
        \end{dcases}
\end{equation}
Here, Regular means the case of regular POVM operators such that $p_\omega (\vec{\theta}) > 0$, while Null means the case of null POVM operators such that $p_\omega (\vec{\theta}) = 0$, where ${\chi}_{\omega,i}(\vec{\Tilde{\theta}}) \equiv \tr [\vr_{\vec{\Tilde{\theta}}} E_\omega L_i (\vec{\Tilde{\theta}})]$. In comparison, using $\sum_\omega E_\omega = \eins$, we can write $\mathcal{Q}(\vr_{\vec{\theta}}) = \sum_{\omega} \mathcal{Q}_\omega(\vr_{\vec{\theta}})$, where $[\mathcal{Q}_\omega(\vr_{\vec{\theta}})]_{ij} = \tr (\vr_{\vec{\theta}} L_i E_\omega L_j )$.

\vspace{1em}
\paragraph*{Appendix~B: Derivation of Eq.~(\ref{eq:classical_GT}).---}
Since the SLD is given by $L_i = 2 ( \ket{\partial_i \psi_{\vec{\theta}}}\! \bra{\psi_{\vec{\theta}}} + \ket{\psi_{\vec{\theta}}}\! \bra{\partial_i \psi_{\vec{\theta}}})$, we write $\chi_{\omega,i}(\vec{\theta}) = 2 \braket{\psi_{\vec{\theta}}| E_\omega| \partial_i \psi_{\vec{\theta}}} + 2 p_\omega (\vec{\theta}) \braket{\partial_i \psi_{\vec{\theta}}|\psi_{\vec{\theta}}}$. Then, $[\chi_{\omega,i}(\vec{\theta})]^* \chi_{\omega,j}(\vec{\theta}) = 4 \braket{\partial_i \psi_{\vec{\theta}}|\mathcal{J}_\omega|\partial_j \psi_{\vec{\theta}}}$, where
\begin{align} \nonumber
    \mathcal{J}_\omega
    &\equiv E_\omega \ket{\psi_{\vec{\theta}}}\! \bra{\psi_{\vec{\theta}}} E_\omega 
    + [p_\omega (\vec{\theta})]^2 \ket{\psi_{\vec{\theta}}}\! \bra{\psi_{\vec{\theta}}}
    \\
    &- p_\omega (\vec{\theta}) (E_\omega \ket{\psi_{\vec{\theta}}}\! \bra{\psi_{\vec{\theta}}} + \ket{\psi_{\vec{\theta}}}\! \bra{\psi_{\vec{\theta}}} E_\omega).
\end{align}
Here we use $\braket{\partial_i \psi_{\vec{\theta}}|\psi_{\vec{\theta}}} \!=\! - \braket{\psi_{\vec{\theta}}|\partial_i \psi_{\vec{\theta}}}$ and $\braket{\psi_{\vec{\theta}}|\psi_{\vec{\theta}}}\!=\!1$. Inserting this into Eq.~(\ref{eq:def_matrixC}) and using $\sum_\omega E_\omega \!=\! \eins$, we obatin Eq.~(\ref{eq:classical_GT}). Similarly, the null-POVM case can be shown.

\vspace{1em}
\paragraph*{Appendix~C: The gauge invariance of the SCGT.---}
For $\ket{\psi_{\vec{\theta}}^\prime} = e^{i \alpha_{\vec{\theta}}} \ket{\psi_{\vec{\theta}}}$, we have $\ket{\partial_i \psi_{\vec{\theta}}^\prime} = (i \partial_i \alpha_{\vec{\theta}}) \ket{\psi_{\vec{\theta}}^\prime} + e^{i \alpha_{\vec{\theta}}} \ket{\partial_i \psi_{\vec{\theta}}}$ and $\mathcal{M}(\ket{\psi_{\vec{\theta}}},\mathsf{E}) = \mathcal{M}(\ket{\psi_{\vec{\theta}}^\prime},\mathsf{E})$. This yields $\braket{\partial_i \psi_{\vec{\theta}}^\prime|\mathcal{M}(\ket{\psi_{\vec{\theta}}^\prime},\mathsf{E})|\partial_j \psi_{\vec{\theta}}^\prime} = \braket{\partial_i \psi_{\vec{\theta}}|\mathcal{M}(\ket{\psi_{\vec{\theta}}},\mathsf{E})|\partial_j \psi_{\vec{\theta}}} + d_{ij}$, where $d_{ij} \equiv (\partial_i \alpha_{\vec{\theta}})(\partial_j \alpha_{\vec{\theta}}) - (i\partial_i \alpha_{\vec{\theta}}) \braket{\psi_{\vec{\theta}}|\partial_j \psi_{\vec{\theta}}} + (i \partial_j \alpha_{\vec{\theta}}) \braket{\partial_i \psi_{\vec{\theta}}|\psi_{\vec{\theta}}}$. Here we used $\braket{\psi_{\vec{\theta}}|\mathcal{M}(\ket{\psi_{\vec{\theta}}},\mathsf{E})|\partial_j \psi_{\vec{\theta}}} = \braket{\psi_{\vec{\theta}}|\partial_j \psi_{\vec{\theta}}}$ and $\braket{\psi_{\vec{\theta}}|\mathcal{M}(\ket{\psi_{\vec{\theta}}},\mathsf{E})|\psi_{\vec{\theta}}} = 1$. Also we have $\braket{\partial_i \psi_{\vec{\theta}}^\prime|\partial_j \psi_{\vec{\theta}}^\prime} = \braket{\partial_i \psi_{\vec{\theta}}|\partial_j \psi_{\vec{\theta}}} + d_{ij}$. Inserting these into Eq.~(\ref{eq:classical_GT}), we find that $\mathcal{C}(\ket{\psi_{\vec{\theta}}^\prime},\mathsf{E}) = \mathcal{C}(\ket{\psi_{\vec{\theta}}},\mathsf{E})$. Similarly, the null-POVM case can be shown.

\vspace{1em}
\paragraph*{Appendix~D: Extension of Eq.~(\ref{eq:ourCQinequality}) to each measurement outcome.---}
Here we show that
\begin{equation} \label{eq:inequality_each_omega}
    \mathcal{C}_\omega(\vr_{\vec{\theta}}) \leq \mathcal{Q}_\omega(\vr_{\vec{\theta}}),
    \quad
    \forall \omega,
\end{equation}
holds, where both terms were considered in Appendix~A. Notice that Eq.~(\ref{eq:ourCQinequality}) is recovered when summing over POVM operators in Eq.~(\ref{eq:inequality_each_omega}). According to Eq.~(\ref{eq:inequality_each_omega}), the saturation condition for the inequality~(\ref{eq:ourCQinequality}) is reduced to that for every outcome, meaning that $\mathcal{C}(\vr_{\vec{\theta}},\mathsf{E}) = \mathcal{Q}(\vr_{\vec{\theta}})$ if and only if $\mathcal{C}_\omega(\vr_{\vec{\theta}}) = \mathcal{Q}_\omega(\vr_{\vec{\theta}})$ for all $\omega$.

For the regular-POVM case, Eq.~(\ref{eq:inequality_each_omega}) is obtained by using the same Cauchy-Schwarz inequality as in the proof of Observation~\ref{ob:ourinequality}, i.e., $|\tr(X^\dagger Y)|^2 \leq \tr(XX^\dagger) \tr(YY^\dagger)$ with $X = \sqrt{E_\omega} \sqrt{\vr_{\vec{\theta}}}$ and $Y = \sqrt{E_\omega} \Tilde{L} \sqrt{\vr_{\vec{\theta}}}$ and $\Tilde{L}=\sum_i z_i L_i$. The necessary and sufficient condition for $\mathcal{C}_\omega(\vr_{\vec{\theta}}) = \mathcal{Q}_\omega(\vr_{\vec{\theta}})$ is the saturation of the Cauchy-Schwarz inequality for all possible choice of $\vec{z}$, i.e., the existence of complex coefficients $\mu_{\omega, i}$ such that
\begin{equation} \label{saturation1}
    E_\omega \vr_{\vec{\theta}} = \mu_{\omega, i} E_\omega L_i \vr_{\vec{\theta}},
    \quad
    \forall i.
\end{equation} 

For the null-POVM case, one first observes that all the eigenvectors of $\vr_{\vec{\theta}}$ lie in the kernel of $E_{\omega}$, i.e., $E_{\omega}\vr_{\vec{\theta}} \!=\!\vr_{\vec{\theta}}E_{\omega}\!=\!0$. Using the observation and the definition of the SLD, a similar manipulation to Ref.~\cite{yang2019optimal} shows that $\partial_{i}p_{\omega}(\vec{\theta})\!=\!{\rm Re}[\tr(L_{i}\vr_{\vec{\theta}}E_{\omega})]\!=\!0$, $\partial_{i}\partial_{j}p_{\omega}(\vec{\theta})\!=\!\{[\mathcal{Q}_{\omega}(\vr_{\vec{\theta}})]_{ij}\!+\![\mathcal{Q}_{\omega}(\vr_{\vec{\theta}})]_{ji}\}/4$, and $\partial_{i}\chi_{\omega,j}({\vec{\theta}})\!=\!(1/2)[\mathcal{Q}_{\omega}(\vr_{\vec{\theta}})]_{ij}$. Inserting these for the Taylor expansions of $p_\omega(\vec{\Tilde{\theta}})$ and ${\chi}_{\omega,j}(\vec{\Tilde{\theta}})$ in Eq.~(\ref{eq:SCGR_omega_level}) yields $p_\omega(\vec{\Tilde{\theta}}) = \delta\vec{\theta}^{T}\mathcal{Q}_{\omega}(\vr_{\vec{\theta}})\delta\vec{\theta}$ and ${\chi}_{\omega,j}(\vec{\Tilde{\theta}}) = \sum_{ij} [\mathcal{Q}_{\omega}(\vr_{\vec{\theta}})]_{ij} \delta\theta_i$, where $\delta\vec{\theta}=\tilde{\vec{\theta}}-\vec{\theta}$. Then,
\begin{equation} \label{eq:appendix_null}
\vec{z}^{\dagger}\mathcal{C}_{\omega}(\vr_{\vec{\theta}})\vec{z}=\frac{\lvert \vec{z}^{\dagger}\mathcal{Q}_{\omega}(\vr_{\vec{\theta}})\delta\vec{\theta} \rvert^{2}}{\delta\vec{\theta}^{T}\mathcal{Q}_{\omega}(\vr_{\vec{\theta}})\delta\vec{\theta}}.
\end{equation}
We apply the Cauchy-Schwarz inequality $|\tr(X^\dagger Y)|^2 \leq \tr(XX^\dagger) \tr(YY^\dagger)$ for the numerator in Eq.~(\ref{eq:appendix_null}). Taking $X = \sqrt{E_{\omega}} \Hat{L} \sqrt{\vr_{\vec{\theta}}}$ and $Y = \sqrt{E_\omega} \Tilde{L} \sqrt{\vr_{\vec{\theta}}}$ with $\Hat{L}=\sum_i \delta \theta_i L_i$ and $\Tilde{L}=\sum_i z_i L_i$, we obtain
\begin{equation}
    \lvert \vec{z}^{\dagger}\mathcal{Q}_{\omega}(\vr_{\vec{\theta}})\delta\vec{\theta} \lvert^{2}
    \leq
    [\vec{z}^{\dagger}\mathcal{Q}_{\omega}(\vr_{\vec{\theta}})\vec{z} ]
    \cdot
    [\delta\vec{\theta}^{T}\mathcal{Q}_{\omega}(\vr_{\vec{\theta}})\delta\vec{\theta} ].
\end{equation}
This directly yields $\mathcal{C}_{\omega}(\vr_{\vec{\theta}}) \! \leq \! \mathcal{Q}_{\omega}(\vr_{\vec{\theta}})$. The necessary and sufficient condition for ${\mathcal{C}}_\omega(\vr_{\vec{\theta}}) \!=\! \mathcal{Q}_\omega(\vr_{\vec{\theta}})$ is the saturation of the Cauchy-Schwarz inequality for all choices of $\vec{z}$, i.e., the existence of complex coefficients $\mu_{\omega, ij}$ such that
\begin{equation} \label{saturation2}
    E_\omega L_i \vr_{\vec{\theta}} = \mu_{\omega, ij} E_\omega L_j \vr_{\vec{\theta}},
    \quad
    \forall i,j.
\end{equation}

\vspace{1em}
\paragraph*{Appendix~E: Saturation of Eq.~(\ref{eq:ourCQinequality}) in the regular-POVM case.---}
Let $E_{\omega} = \sum_\alpha e_{\omega, \alpha} \ket{\pi_{\omega, \alpha}} \! \bra{\pi_{\omega, \alpha}}$ be the spectral decomposition of a general POVM element, with $e_{\omega, \alpha}\geq 0$ and the eigenstates $\ket{\pi_{\omega, \alpha}}$ being not necessarily orthogonal for different $\omega$. Considering the spectral decomposition $\vr_{\vec{\theta}}=\sum_{x} \lambda_{x,\vec{\theta}} \ket{\psi_{x, \vec{\theta}}} \! \bra{\psi_{x, \vec{\theta}}}$, with $\lambda_{x,\vec{\theta}} \geq 0$, we can thus rewrite Eq.~(\ref{saturation1}) as 
\begin{equation} \label{endmatterEq}
    \sum_{x, \alpha}
    \lambda_{x,\vec{\theta}} e_{\omega, \alpha}
    \braket{\pi_{\omega, \alpha}|\eins \!-\! \mu_{\omega, i} L_i|\psi_{x, \vec{\theta}}}
    \ket{\pi_{\omega, \alpha}} \! \bra{\psi_{x, \vec{\theta}}} \!=\! 0.
\end{equation}

Due to the linear independence among the set of operators $\{\ket{\pi_{\omega, \alpha}} \! \bra{\psi_{x, \vec{\theta}}}\}_{x,\alpha}$, Eq.~(\ref{endmatterEq}) can be fulfilled if and only if each term in the bracket is equal to zero. This can be seen more explicitly by projecting on the right and left side of Eq.~(\ref{endmatterEq}) over a complete basis. In other words, the condition Eq.~(\ref{endmatterEq}) is equivalent to asking the corresponding matrix to be null for any possible choice of basis, which is only possible if the bracket term vanishes.

Without loss of generality, we can restrict to rank-one POVM operator $E_{\omega}=\ket{\pi_{\omega}}\! \bra{\pi_{\omega}}$. The necessary and sufficient condition is therefore the existence of a coefficient $\mu_{\omega, i}$ such that $\braket{\pi_{\omega}|\psi_{x, \vec{\theta}}} = \mu_{\omega, i} \braket{\pi_{\omega}|L_i|\psi_{x, \vec{\theta}}}$, for all $i$ and $x$. This condition is equal to $\braket{\pi_\omega |L_i|\psi_{x, \vec{\theta}}} /\braket{\pi_\omega |\psi_{x, \vec{\theta}}} = \braket{\pi_\omega |L_i|\psi_{y, \vec{\theta}}} /\braket{\pi_\omega |\psi_{y, \vec{\theta}}}$ for all $i$ and $x, y$, which leads to Eq.~(\ref{eq:R_case}) using $\tr(A) \tr(B) = \tr(A \otimes B)$.

\vspace{1em}
\paragraph*{Appendix~F: Saturation of Eq.~(\ref{eq:ourCQinequality}) in the null-POVM case.---}
Following Appendix~E, we consider the spectral decomposition of $E_\omega$ and $\vr_{\vec{\theta}}$. It is then suffices to consider rank-one null POVM operators such that Eq.~(\ref{saturation2}) becomes equivalent to $\braket{\pi_\omega |L_i|\psi_{x, \vec{\theta}}} = \mu_{\omega, ij} \braket{\pi_\omega |L_j|\psi_{x, \vec{\theta}}}$ for all $i,j$ and $x$. This can be also rewritten as $\braket{\pi_\omega |L_i|\psi_{x, \vec{\theta}}}/\braket{\pi_\omega |L_j|\psi_{x, \vec{\theta}}} = \braket{\pi_\omega |L_i|\psi_{y, \vec{\theta}}}/\braket{\pi_\omega |L_j|\psi_{y, \vec{\theta}}}$ for all $i,j$ and $x, y$. It immediately leads to the necessary and sufficient condition
\begin{equation}
    \label{eq:N_case}
    \bra{\pi_\omega} \otimes \bra{\pi_\omega} (L_i \otimes L_j - L_j \otimes L_i)
    \ket{\psi_{x, \vec{\theta}}} \otimes \ket{\psi_{y, \vec{\theta}}}
    = 0,
\end{equation}
for all $i,j,\omega,x,y$.

\vspace{1em}
\paragraph*{Appendix~G: Expressions for pure states and unitary transformations.---}
Here we present the explicit expressions for the key quantities discussed in this manuscript, considering the simple case of a pure state $\ket{\psi_{\vec{\theta}}} = U_{\vec{\theta}} \ket{\psi}$, where $U_{\vec{\theta}}$ is a unitary parameter-encoding transformation. {By recalling that $\mathcal{Q} = \mathcal{F}_Q + i \mathcal{G}$ and $\mathcal{C} = \mathcal{F}_C + \mathcal{I} + i \mathcal{D}$, a direct calculation yields
\begin{subequations}
    \begin{align}
        [\mathcal{Q}]_{ij}
        &\!=\!
        4 \left[ \braket{\mathcal{H}_i \mathcal{H}_j}
        \!-\! \braket{\mathcal{H}_i} \! \braket{\mathcal{H}_j} \right],
        \\
        [\mathcal{F}_Q]_{ij}
        &\!=\! 2 \braket{\mathcal{H}_i \mathcal{H}_j \!+\! \mathcal{H}_j \mathcal{H}_i}\!-\! 4 \braket{\mathcal{H}_i} \! \braket{\mathcal{H}_j},
        \\
        [\mathcal{G}]_{ij}
        &\!=\! -2i \braket{\mathcal{H}_i \mathcal{H}_j \!-\! \mathcal{H}_j \mathcal{H}_i},
        \\
        [\mathcal{C}]_{ij}
        &\!=\!
        4 \left[ \braket{\mathcal{H}_i \mathcal{N} \mathcal{H}_j}
        \!-\! \braket{\mathcal{H}_i} \! \braket{\mathcal{H}_j} \right],
        \\
        \! \! \! \!
        [\mathcal{F}_C + \mathcal{I}]_{ij}
        &\!=\! 2 \braket{\mathcal{H}_i \mathcal{N} \mathcal{H}_j
        \!+\! \mathcal{H}_j \mathcal{N} \mathcal{H}_i} 
        \!-\! 4 \braket{\mathcal{H}_i} \! \braket{\mathcal{H}_j},
        \\
        [\mathcal{D}]_{ij}
        &\!=\! -2i 
        \braket{\mathcal{H}_i \mathcal{N} \mathcal{H}_j \!-\! \mathcal{H}_j \mathcal{N}\mathcal{H}_i},  
        \\
        [\mathcal{F}_C]_{ij}
        &\!=\! \braket{\mathcal{H}_i \mathcal{N} \mathcal{H}_j \!+\! \mathcal{H}_j \mathcal{N} \mathcal{H}_i} \!-\! \eta_{ij},
        \\
        [\mathcal{I}]_{ij}
        &\!=\! \braket{\mathcal{H}_i \mathcal{N} \mathcal{H}_j \!+\! \mathcal{H}_j \mathcal{N} \mathcal{H}_i}
        \!-\! 4 \braket{\mathcal{H}_i} \! \braket{\mathcal{H}_j}
        \!+\! \eta_{ij},      
    \end{align}
\end{subequations}
where the dependencies on $\vec{\theta}$, the state, and the POVM are omitted here and below. In the above expressions, we denoted that $\braket{X} \equiv \braket{\psi|X|\psi}$ for an operator $X$ and $\mathcal{H}_i \equiv - i (\partial_i U_{\vec{\theta}}^\dagger) U_{\vec{\theta}}$. Here, $\mathcal{N} \equiv \mathcal{N} (\ket{\psi_{\vec{\theta}}},\mathsf{E})$ and $\eta_{ij} \equiv \eta_{ij} (\ket{\psi_{\vec{\theta}}},\mathsf{E}) = \braket{\psi, \psi|\mathcal{O}_{ij}|\psi, \psi}$ with 
\begin{subequations}
\begin{align}
    \mathcal{N}
    &\!=\! U_{\vec{\theta}}^\dagger \mathcal{M}(\ket{\psi_{\vec{\theta}}},\mathsf{E}) U_{\vec{\theta}}
    = \sum_\omega \frac{1}{p_\omega (\vec{\theta})}
    \mathcal{E}_\omega \ket{\psi}\! \bra{\psi} \mathcal{E}_\omega,
    \\
    \mathcal{O}_{ij}
    &\!=\! \sum_\omega \frac{1}{p_\omega (\vec{\theta})}
    (\mathcal{E}_\omega \mathcal{H}_i \otimes \mathcal{E}_\omega \mathcal{H}_j + \mathcal{H}_i \mathcal{E}_\omega \otimes \mathcal{H}_j \mathcal{E}_\omega ),
\end{align}
\end{subequations}
where $\ket{\psi,\psi} = \ket{\psi \otimes \psi}$, $\mathcal{E}_\omega = U_{\vec{\theta}}^\dagger E_\omega U_{\vec{\theta}}$, and $\mathcal{M}(\ket{\psi_{\vec{\theta}}},\mathsf{E})$ is defined in Eq.~(\ref{eq:classical_GT}). For a rank-one POVM $\mathsf{E} = \{\ket{\pi_\omega}\! \bra{\pi_\omega}\}$, we have $\mathcal{M}(\ket{\psi_{\vec{\theta}}},\mathsf{E}) = \eins$ and thus $\mathcal{N} = \eins$.
}

\vspace{1em}
\paragraph*{Appendix~H: Necessary and sufficient condition for $\mathcal{I}(\vr_{\vec{\theta}},\mathsf{E})=0$.---}
Similarly to Appendix~A, let us write $\mathcal{I}(\vr_{\vec{\theta}},\mathsf{E}) = \sum_{\omega} \mathcal{I}_\omega(\vr_{\vec{\theta}})$, where $\mathcal{I}_\omega(\vr_{\vec{\theta}})$ has elements 
\begin{equation} \nonumber
    [\mathcal{I}_\omega(\vr_{\vec{\theta}})]_{ij} \!=\! 
        \begin{dcases}
        \frac{
        {\rm Im}[\chi_{\omega,i}(\vec{\theta})]
        {\rm Im}[\chi_{\omega,j}(\vec{\theta})]
        }{ p_\omega (\vec{\theta}) }
        \!
        &{\rm for} \,\, {\rm Regular}\\ 
        \lim_{\vec{\Tilde{\theta}} \to \vec{\theta}}
        \frac{
        {\rm Im}[\chi_{\omega,i}(\vec{\Tilde{\theta}})]
        {\rm Im}[\chi_{\omega,j}(\vec{\Tilde{\theta}})]
        }{p_\omega(\vec{\Tilde{\theta}})}
        \!
        &{\rm for} \,\, {\rm Null}\\
        \end{dcases}
\end{equation}
The necessary and sufficient condition for $\mathcal{I}(\vr_{\vec{\theta}},\mathsf{E})=0$ is given by $\mathcal{I}_\omega(\vr_{\vec{\theta}}) = 0$ for all $\omega$, since $\mathcal{I}_\omega(\vr_{\vec{\theta}})$ are positive-semidefinite matrices.

For a rank-one regular POVM $E_\omega = \ket{\pi_\omega}\! \bra{\pi_\omega}$, the condition ${\rm Im}[\chi_{\omega,i}(\vec{\theta})] = 0$ for all $i$ is equivalent to
\begin{equation} \label{eq:regular_POVM_Izero}
    \braket{\pi_\omega|L_i \vr_{\vec{\theta}} - \vr_{\vec{\theta}} L_i|\pi_\omega} = 0,
    \quad
    \forall i.
\end{equation}
For pure states, using $L_i = 2 ( \ket{\partial_i \psi_{\vec{\theta}}}\! \bra{\psi_{\vec{\theta}}} + \ket{\psi_{\vec{\theta}}}\! \bra{\partial_i \psi_{\vec{\theta}}})$, we can thus rewrite Eq.~(\ref{eq:regular_POVM_Izero}) as ${\rm Im}[\braket{\partial_i \psi_{\vec{\theta}}|\pi_\omega} \! \braket{\pi_\omega|\psi_{\vec{\theta}}}] = \lvert \braket{\psi_{\vec{\theta}}|\pi_\omega} \rvert^2 {\rm Im}[\braket{\partial_i \psi_{\vec{\theta}}|\psi_{\vec{\theta}}}]$, which is equivalent to the necessary and sufficient condition for $\mathcal{F}_C(\ket{\psi_{\vec{\theta}}},\mathsf{E})=\mathcal{F}_Q(\ket{\psi_{\vec{\theta}}})$ presented in Eq.~(8) of Ref.~\cite{pezze2017optimal}. For mixed states, Eq.~(\ref{eq:regular_POVM_Izero}) becomes equivalent to the existence of real coefficients $\mu_{\omega, i}$ such that $\braket{\pi_{\omega}|\psi_{x, \vec{\theta}}} = \mu_{\omega, i} \braket{\pi_{\omega}|L_i|\psi_{x, \vec{\theta}}}$. This recovers the necessary and sufficient condition for $\mathcal{F}_C(\vr_{\vec{\theta}},\mathsf{E})=\mathcal{F}_Q(\vr_{\vec{\theta}})$ presented in Eq.~(39) of Ref.~\cite{yang2019optimal}.

For a rank-one null POVM $E_\omega = \ket{\pi_\omega}\! \bra{\pi_\omega}$, the condition ${\rm Im}[{\chi}_{\omega,i}(\vec{\Tilde{\theta}})] = 0$ for all $i$ is equivalent to
\begin{equation} \label{eq:null_POVM_Izero}
    \braket{\pi_\omega|L_i \vr_{\vec{\theta}} L_j - L_j \vr_{\vec{\theta}} L_i
    |\pi_\omega}
    = 0,
    \quad
    \forall i,j,
\end{equation}
where we used that ${\chi}_{\omega,j}(\vec{\Tilde{\theta}}) = \sum_{ij} [\mathcal{Q}_{\omega}(\vr_{\vec{\theta}})]_{ij} \delta\theta_i$ given in Appendix~D with $[\mathcal{Q}_\omega(\vr_{\vec{\theta}})]_{ij} = \tr (\vr_{\vec{\theta}} L_i E_\omega L_j )$. For pure states, we can thus rewrite Eq.~(\ref{eq:null_POVM_Izero}) as ${\rm Im}[\braket{\partial_i \psi_{\vec{\theta}}|\pi_\omega} \! \braket{\pi_\omega|\partial_j \psi_{\vec{\theta}}}]=0$, which is equivalent to Eq.~(7) of Ref.~\cite{pezze2017optimal}. For mixed states, Eq.~(\ref{eq:null_POVM_Izero}) becomes equivalent to the existence of real coefficients $\mu_{\omega, ij}$ such that $\braket{\pi_\omega |L_i|\psi_{x, \vec{\theta}}} = \mu_{\omega, ij} \braket{\pi_\omega |L_j|\psi_{x, \vec{\theta}}}$. This recovers the previous condition presented in Eq.~(44) of Ref.~\cite{yang2019optimal}.

\vspace{1em}
\paragraph*{Appendix~I: Practical accessibility of the SCGT.---}
Here we discuss the practical accessibility of the SCGT for a pure state $\ket{\psi_{\vec{\theta}}} = U_{\vec{\theta}} \ket{\psi}$. First, according to Appendix~G, we have
\begin{subequations}
\begin{align}
    \braket{\mathcal{H}_i \mathcal{N} \mathcal{H}_j \!+\! \mathcal{H}_j \mathcal{N} \mathcal{H}_i}
    &= \braket{\psi, \psi| \mathcal{K}_{ij}^+ |\psi, \psi},
    \label{eq:two:K+}
    \\
    (-i) \braket{\mathcal{H}_i \mathcal{N} \mathcal{H}_j \!-\! \mathcal{H}_j \mathcal{N} \mathcal{H}_i}
    &= \braket{\psi, \psi| \mathcal{K}_{ij}^-|\psi, \psi},
    \label{eq:two:K-}
\end{align}    
\end{subequations}
where $\mathcal{K}_{ij}^{+} \!=\! \sum_{\omega} [1/p_\omega (\vec{\theta})] (\mathcal{H}_i \mathcal{E}_\omega \otimes \mathcal{E}_\omega \mathcal{H}_j + \mathcal{H}_j \mathcal{E}_\omega \otimes \mathcal{E}_\omega \mathcal{H}_i)$, and $\mathcal{K}_{ij}^{-} \!=\! \sum_{\omega} [1/p_\omega (\vec{\theta})]  (-i) (\mathcal{H}_i \mathcal{E}_\omega \otimes \mathcal{E}_\omega \mathcal{H}_j - \mathcal{H}_j \mathcal{E}_\omega \otimes \mathcal{E}_\omega \mathcal{H}_i)$. Note that $\mathcal{K}_{ij}^{\pm} = \mathcal{K}_{ij}^{\pm} (\ket{\psi_{\vec{\theta}}}, \mathsf{E})$ and its state dependency arises through $p_\omega (\vec{\theta})$.

Therefore, given two copies of any unknown state $\ket{\psi}$ and a known unitary $U_{\vec{\theta}}$, the SCGT associated with any POVM can be experimentally accessed by measuring the following quantities: $p_{\omega}(\vec{\theta})$, $\braket{\psi|\mathcal{H}_i|\psi}$, and $\braket{\psi, \psi| \mathcal{K}_{ij}^{\pm} |\psi, \psi}$. Now, $\mathcal{K}_{ij}^{\pm}$ are non-Hermitian operators, while their expectation values in the two-copy system can be measured via Hermitian operators as follows: For non-Hermitian operators $A$, $B$, and $\mathcal{X}_+ = A \otimes B + B^\dagger \otimes A^\dagger$ and $\mathcal{X}_- = (-i) (A \otimes B - B^\dagger \otimes A^\dagger)$, it holds that
\begin{subequations}
\begin{align}
    \! \! \! 
    \braket{\psi, \psi|\mathcal{X}_+|\psi, \psi}
    &\!=\! \braket{\psi, \psi|\frac{A_+ \! \otimes \! B_+ \!-\! A_- \! \otimes \! B_-}{2}|\psi, \psi},
    \label{eq:two-copy:+}
    \\
    \! \! \! 
    \braket{\psi, \psi|\mathcal{X}_-|\psi, \psi}
    &\!=\! \braket{\psi, \psi|\frac{A_+ \! \otimes \! B_- \!+\! A_- \! \otimes \! B_+}{2}|\psi, \psi},
    \label{eq:two-copy:-}
\end{align}    
\end{subequations}
where the Hermitian elements of $X = A,B$ are $X_+ = (1/2) (X + X^\dagger)$ and $X_- = (1/2i) (X - X^\dagger)$. Hence, taking $\mathcal{X}_{\pm} = \mathcal{K}_{ij}^{\pm}$ and provided that a two-copy state is available, one can in principle construct the Hermitian observables such as the right-hand side in Eqs.~(\ref{eq:two-copy:+},\ref{eq:two-copy:-}) to obtain Eqs.~(\ref{eq:two:K+},\ref{eq:two:K-}).

\end{document}